\documentclass[11pt]{article}
\usepackage{graphicx,amssymb,amsmath}
\usepackage{epsf}
\usepackage{amsfonts}
\usepackage{latexsym}
\usepackage{amssymb}
\usepackage{amsthm}
\usepackage{graphics}
\usepackage{psfrag}
\usepackage{float}
\usepackage{bbm}
\usepackage[small]{caption}

\usepackage{paralist}%
\usepackage{xspace}%
\usepackage{xcolor}%
\usepackage{euscript}%

\newtheorem{theorem}{Theorem}
\newtheorem{lemma}{Lemma}

\newcommand{\old}[1]{{}}
\newcommand{\later}[1]{{}}

\def\etal{{et~al.}}

\def\ie{{i.e.}}

\newcommand{\alg}{{\rm ALG}}
\newcommand{\opt}{{\rm OPT}}

\newcommand{\eps}{\varepsilon}
\newcommand{\NN}{\mathbb{N}}
\newcommand{\ZZ}{\mathbb{Z}}
\newcommand{\RR}{\mathbb{R}}

\newcommand{\conv}{{\rm conv}}

\newcommand{\PntSet}{S}%
\newcommand{\PntSetA}{Q}%
\newcommand{\Body}{C}%
\newcommand{\Copt}{C_{\mathrm{opt}}}%
\newcommand{\VolX}[1]{\mathrm{vol}\pth{#1}}%

\newcommand{\pnt}{p}%

\newcommand{\pth}[2][\!]{#1\left({#2}\right)}
\newcommand{\ceil}[1]{\left\lceil {#1} \right\rceil}

\newcommand{\lemlab}[1]{\label{lemma:#1}}
\newcommand{\lemref}[1]{Lemma~\ref{lemma:#1}}

\newcommand{\seclab}[1]{\label{sec:#1}}
\newcommand{\secref}[1]{Section~\ref{sec:#1}}

\newcommand{\Grid}{\mathcal{G}}
\newcommand{\GridX}[1]{\Grid\pth{#1}}
\newcommand{\brc}[1]{\left\{ {#1} \right\}}
\newcommand{\sep}[1]{\,\left|\, {#1} \MakeBig\right.}
\newcommand{\MakeBig}{\rule[-.2cm]{0cm}{0.4cm}}

\definecolor{blue25}{rgb}{0,0,0.95}

\newcommand{\bd}{{\partial}}%%
\newcommand{\cardin}[1]{\left|{#1}\right|}%

\newcommand{\thmlab}[1]{{\label{theo:#1}}}
\newcommand{\thmref}[1]{Theorem~\ref{theo:#1}}

\newcommand{\figlab}[1]{\label{fig:#1}}
\newcommand{\figref}[1]{Figure~\ref{fig:#1}}

\title{Minimum Convex Partitions and Maximum Empty Polytopes\footnote{A
preliminary version of this paper appeared in the {\em Proceedings of
the 13th Scandinavian Symposium and Workshops on Algorithm Theory,
Helsinki, Finland, July, 2012.}}}

\author{%
   Adrian Dumitrescu%
   \thanks{Department of Computer Science, University of
      Wisconsin--Milwaukee, WI 53201-0784, USA\@.
      Email:~\texttt{dumitres@uwm.edu}.  Supported in part by NSF
      grant DMS-1001667.}%
   \and%
   Sariel Har-Peled%
   \thanks{Department of Computer Science, University of Illinois at
      Urbana--Champaign, Urbana, IL 61801-2302, USA\@.
      Email:~\texttt{sariel@cs.uiuc.edu}. Work on this paper was
      partially supported by a NSF AF award CCF-0915984.}%
   \and%
   Csaba D. T\'oth\thanks{Department of Mathematics,
      California State University, Northridge, Los Angeles, CA; and
      Department of Computer Science, Tufts University, Medford, MA, USA\@.
      Email:~\texttt{cdtoth@acm.org}.  Work on this paper was
      supported in part by NSERC grant RGPIN 35586.}%
}

\begin{document}
\maketitle

\begin{abstract}
    Let $S$ be a set of $n$ points in $\RR^d$. A Steiner convex
    partition is a tiling of $\conv(S)$ with empty convex bodies.  For
    every integer $d$, we show that $S$ admits a Steiner convex
    partition with at most $\lceil (n-1)/d\rceil$ tiles. This bound is
    the best possible for points in general position in the plane, and
    it is best possible apart from constant factors in every fixed
    dimension $d\geq 3$. We also give the first constant-factor
    approximation algorithm for computing a minimum Steiner convex
    partition of a planar point set in general position.

    Establishing a tight lower bound for the maximum volume of a tile
    in a Steiner convex partition of any $n$ points in the unit cube is
    equivalent to a famous problem of Danzer and Rogers.  It is
    conjectured that the volume of the largest tile is $\omega(1/n)$.
    Here we give a $(1-\eps)$-approximation
    algorithm for computing the maximum volume of an empty convex body
    amidst $n$ given points in the $d$-dimensional unit box $[0,1]^d$.


\end{abstract}

\section{Introduction} \seclab{sec:intro}

Let $S$ be a set of $n\geq d+1$ points in $\RR^d$, $d \geq 2$.  A
convex body $C$ is \emph{empty} if its interior is disjoint from $S$.
A \emph{convex partition} of $S$ is a partition of the convex hull
$\conv(S)$ into empty convex bodies (called \emph{tiles}) such that
the vertices of the tiles are in $S$. In a \emph{Steiner convex
   partition} of $S$ the vertices of the tiles are arbitrary: they can
be points in $S$ or \emph{Steiner points}.  For instance, any
triangulation of $S$ is a convex partitions of $S$, where the convex
bodies are simplices, and so $\conv(S)$ can be always partitioned into
$O(n^{\lfloor d/2\rfloor})$ empty convex tiles~\cite{DRS10}.

In this paper, we study the minimum number of tiles that a Steiner
convex partition of every $n$ points in $\RR^d$ admits, and the
maximum volume of a single tile for a given point set.  The research
is motivated by a longstanding open problem by Danzer and
Rogers~\cite{ABFK92, BW71, BC87, FP94, PT12}:
What is the maximum volume of an empty convex body $C\subset [0,1]^d$
that can be found amidst any set $S\subset [0,1]^d$ of $n$ points in a
unit cube?  The current best bounds are $\Omega(1/n)$ and
$O(\log{n}/n)$, respectively (for a fixed $d$).  The lower bound,
for instance, can be deduced by decomposing the unit cube by
$n$~parallel hyperplanes, each containing at least one point, into at
most $n+1$ empty convex bodies.
The upper bound is tight apart from constant factors for $n$
randomly and uniformly distributed points in the unit cube.
It is suspected that the
largest volume is $\omega(1/n)$ in any dimension $d \geq 2$, \ie, the
ratio between this volume and $1/n$ tends to $\infty$.

For a convex body $C$ in $\RR^d$, denote by $\VolX{C}$ the Lebesgue
measure of $C$, \ie, its area when $d = 2$, or its volume when $d \ge 3$.

\paragraph{Minimum number of tiles in a convex partition.}
A {\em minimum convex partition} of $S$ is a convex
partition of $S$ with a minimum number of tiles. Denote this number by
$f_d(S)$.  Further define (by slightly abusing notation)
$$f_d(n)=\max \{f_d(S): S \subset \RR^d, |S|=n\}. $$
Similarly define a {\em minimum Steiner convex partition} of $S$ as
one with a minimum number of tiles and let $g_d(S)$ denote this
number.  We also define
$$ g_d(n)=\max \{g_d(S): S \subset \RR^d, |S|=n\}. $$

There has been substantial work on estimating $f_2(n)$, and computing
$f_2(S)$ for a given set $S$ in the plane.  It has been shown
successively that $f_2(n) \leq \frac{10n-18}{7}$ by
Neumann-Lara~\etal~\cite{NRU04}, $f_2(n) \leq \frac{15n-24}{11}$ by
Knauer and Spillner~\cite{KS06}, and $f_2(n)\leq \frac{4n-6}{3}$ for
$n\geq 6$ by Sakai and Urrutia~\cite{SU09}.  From the other direction,
Garc\'ia-L\'opez and Nicol\'as~\cite{GN06} proved that $f_2(n) \geq
\frac{12n-22}{11}$, for $n\geq 4$, thereby improving an earlier lower
bound $f_2(n) \geq n+2$ by Aichholzer and Krasser~\cite{AK01}.  Knauer
and Spillner~\cite{KS06} have also obtained a $\frac{30}{11}$-factor
approximation algorithm for computing a minimum convex partition for a
given set $S\subset \RR^2$, no three of which are collinear. There are
also a few exact algorithms, including three fixed-parameter
algorithms~\cite{FMR01,GL05,Sp08}.

The state of affairs is much different in regard to Steiner convex
partitions.  As pointed out in~\cite{DT11}, no corresponding results
are known for the variant with Steiner points. Here we take the first
steps in this direction, and obtain the following results.

\begin{theorem}%
    \thmlab{T1}%
    For $n\geq d+1$, we have $g_d(n)\leq
    \left\lceil\frac{n-1}{d}\right\rceil$.  For $d=2$, this bound is
    the best possible, that is, $g_2(n)=\lceil (n-1)/2\rceil$; and for
    every fixed $d\geq 2$, we have $g_d(n) =\Omega(n)$.
\end{theorem}

We say that a set of points in $\RR^d$ is in \emph{general position}
if every $k$-dimensional affine subspace contains at most $k+1$ points
for $0\leq k< d$.  We show that in the plane every Steiner convex
partition for $n$ points in general position, $i$ of which lie in the
interior of the convex hull, has $\Omega(i)$ tiles. This
leads to a simple constant-factor approximation algorithm.

\begin{theorem} \thmlab{T2}
    Given a set $S$ of $n$ points in general position in the plane, a
    ratio $3$ approximation of a minimum Steiner convex partition of
    $S$ can be computed in $O(n \log{n})$ time.
\end{theorem}

The \emph{average} volume of a tile in a Steiner convex partition of $n$
points in the unit cube $[0,1]^d$ is an obvious lower bound for the
maximum possible volume of a tile, and for the maximum volume of any
empty convex body $C\subset [0,1]^d$. The lower bound $g_d(n)=\Omega(n)$
in \thmref{T1} shows that the average volume of a tile is
$O(1/n)$ in some instances, where the constant of proportionality
depends only on the dimension.  This implies that a simple
``averaging'' argument is not a viable avenue for finding a solution
to the problem of Danzer and Rogers.

\paragraph{Maximum empty polytope among $n$ points in a
   unit cube.}  In the second part of the paper, we consider the
following problem: Given a set of $n$ points in a rectangular box $B$ in
$\RR^d$, find a maximum-volume empty convex body $C \subset B$. Since
the ratio between volumes is invariant under affine transformations,
we may assume without loss of generality that $B=[0,1]^d$.  We
therefore have the problem of computing a maximum volume empty convex
body $C\subset [0,1]^d$ for a set of $n$ points in $[0,1]^d$.
It can be argued that the maximum volume empty convex body is a polytope,
however, the number and location of its vertices is unknown and this
represents the main difficulty.
For $d=2$ there is a polynomial-time exact algorithm (see~\secref{sec:conclusion})
while for $d \geq 3$ we are not aware of any exact algorithm.
Thus the problem of finding faster approximations naturally suggests itself.

There exist exact algorithms for some related problems.
Eppstein~\etal~\cite{EORW92} find the maximum area empty convex
$k$-gon with vertices among $n$ points in $O(kn^3)$ time, if it exists.
As a byproduct, a maximum area empty convex polygon with vertices among $n$
given points can be computed exactly in $O(n^4)$ time with their dynamic
programming algorithm. The running time was subsequently improved to $O(n^3\log n)$
by Fischer~\cite{Fis97} and then to $O(n^3)$  by Bautista-Santiago~\etal~\cite{BDL+11}.

By John's ellipsoid theorem~\cite{Ma02}, the maximum volume empty
ellipsoid in $[0,1]^d$ gives a $1/d^d$-appro\-xi\-ma\-tion.
Here we present a $(1-\eps)$-approximation for a maximum volume empty convex
body $\Copt$ by first guessing a good approximation of the bounding
hyperrectangle of $\Copt$ of minimum volume,
and then finding a sufficiently close approximation of $\Copt$ inside it.
We obtain the following two approximation algorithms.
The planar algorithm runs in quadratic time in $n$,
however, the running time degrades with the dimension.

\begin{theorem}\thmlab{algorithm:2d}
    Given a set $\PntSet$ of $n$ points in $[0,1]^2$ and
    parameter $\eps>0$, one can compute an empty convex body
    $C \subseteq [0,1]^2$, such that $\VolX{C} \geq
    (1-\eps)\VolX{\Copt}$. The running time of the algorithm is
    $O\pth{\eps^{-6} n^2}$.
\end{theorem}

\begin{theorem}\thmlab{empty}
    Given a set $\PntSet$ of $n$ points in $[0,1]^d$, $d\geq 3$, and
    a parameter $\eps>0$, one can compute an empty convex body
    $C \subseteq [0,1]^d$, such that $\VolX{C} \geq
    (1-\eps)\VolX{\Copt}$. The running time of the algorithm is
   $$\exp\pth{ O  \pth{\eps^{-d(d-1)/(d+1)}\log \eps^{-1}}} n^{1 + d(d-1)/2} \log^d n.$$
\end{theorem}

As far as the problem of Danzer and Rogers is concerned,
one need not consider convex sets---it suffices to consider
simplices---and for simplices the problems considered are much
simpler. Specifically, every convex body $C$ in $\RR^d$,
$d \geq 2$, contains a simplex $T$ of volume
$\VolX{T}\geq \VolX{C}/(d+2)^d$~\cite{Las11}.
That is, for fixed $d$, the largest empty simplex amidst $n$ points in
the unit cube $[0,1]^d$ yields a constant-factor approximation
of the largest volume convex body (polytope) amidst the same $n$ points.
Consequently, the asymptotic dependencies on $n$ of the volumes of the
largest empty simplex and convex body are the same.
For $d=2$ there is a polynomial-time exact algorithm for computing the
largest empty triangle amidst $n$ points in $[0,1]^2$ (see~\secref{sec:conclusion})
while for $d \geq 3$ we are not aware of any exact algorithm for computing the
largest empty simplex amidst $n$ points in $[0,1]^d$.

\paragraph{Related work.}
Decomposing polygonal domains into convex sub-polygons has been also
studied extensively. We refer to the article by Keil~\cite{K00} for a
survey of results up to the year 2000. For instance, when the polygon
may contain holes, obtaining a minimum convex partition is NP-hard,
regardless of whether Steiner points are allowed.  For polygons
without holes, Chazelle and Dobkin~\cite{CD79} obtained an $O(n+r^3)$
time algorithm for the problem of decomposing a polygon with $n$
vertices, $r$ of which are reflex, into convex parts, with Steiner
points permitted. Keil~\cite{K00} notes that although there are
an infinite number of possible locations for the Steiner points,
a dynamic programming approach is amenable to obtain an exact
(optimal) solution; see also~\cite{KS02,Sh92}.

Fevens~\etal~\cite{FMR01} designed a polynomial time algorithm for
computing a minimum convex partition for a given set of $n$ points in
the plane if the points are arranged on a constant number of convex
layers.
The problem of minimizing the total Euclidean length of the edges of a
convex partition has been also considered.  Grantson and
Levcopoulos~\cite{GN06}, and Spillner~\cite{Sp08} proved that the
shortest convex partition and Steiner convex partition problems are
fixed parameter tractable, where the parameter is the number of points
of $P$ lying in the interior of $\conv(P)$.
Dumitrescu~and~T\'oth~\cite{DT11} proved that every set of $n$ points in $\RR^2$
admits a Steiner convex partition which is at most $O(\log n/\log \log
n)$ times longer than the minimum spanning tree, and this bound cannot
be improved. Without Steiner points, the best upper bound for the
ratio of the minimum length of a convex partition and the length of a
minimum spanning tree (MST) is $O(n)$~\cite{Ki80}.

A largest area convex polygon contained in a given (non-convex) polygon
with $n$ vertices can be found by the  algorithm of Chang and
Yap~\cite{CY86} in $O(n^7)$ time. The problem is known as the
\emph{potato-peeling problem}.
On the other hand, a largest area triangle contained in a simple polygon
with $n$ vertices, can be found by the  algorithm of Melissaratos and
Souvaine~\cite{ms-sphsg-92} in $O(n^4)$ time.
Hall-Holt~\etal~\cite{hkkms-flspp-06} compute a constant approximation
in time $O(n \log{n})$. The same authors show how to compute a
$(1-\eps)$-approximation of the largest \emph{fat} triangle inside a
simple polygon (if it exists) in time $O(n)$. Given a triangulated
polygon (with possible holes) with $n$ vertices, Aronov~\etal~\cite{AKLS11}
compute the largest area  convex polygon respecting the triangulation edges
in $O(n^2)$ time.

For finding a maximum volume empty axis-parallel box amidst $n$ points
in $[0,1]^d$, Backer and Keil~\cite{BK10} reported an algorithm with
worst-case running time of $O(n^d \log^{d-2} {n})$. An empty
axis-aligned box whose volume is at least $(1-\eps)$ of the maximum
can be computed in
$O\left(\left( \frac{8 e d}{\eps^2} \right)^d \, n \, \log^d{n}\right)$
time by the algorithm of Dumitrescu and Jiang~\cite{DJ12}.

Lawrence and Morris~\cite{LM09} studied the minimum integer $k_d(n)$
such that the complement $\RR^d\setminus S$ of any $n$-element set
$S\subset \RR^d$, not all in a hyperplane, can be \emph{covered} by
$k_d(n)$ convex sets.  They prove $k_d(n) =\Omega(\log n/ d \log
\log n)$.
%Bounds on $k_d(n)$ are also related to the \emph{invisibility graph}
%of a point set~\cite{CKM10}.
It is known that covering the complement of $n$ uniformly distributed points
in $[0,1]^d$ requires $\Omega(n/d\log n)$ convex sets, which
follows from the upper bound in the problem of Danzer and Rogers.

\section{Combinatorial bounds} \seclab{sec:bounds}

In this section we prove \thmref{T1}. We start with the upper bound.
The following simple algorithm returns a Steiner convex partition with
at most $\lceil (n-1)/d\rceil$ tiles for any $n$ points in $\RR^d$.

\medskip
\noindent Algorithm {\bf A1}:
\begin{enumerate} \itemsep 1pt
    \item[{\sc Step 1.}] Compute the convex hull $R \gets \conv(S)$ of
    $S$.  Let $A\subseteq S$ be the set of hull vertices, and let $B
    =S \setminus A$ denote the remaining points.
    \item[{\sc Step 2.}] Compute $\conv(B)$, and let $H$ be the
    supporting hyperplane of an arbitrary $(d-1)$-dimensional face of
    $\conv(B)$.  Denote by $H^+$ the halfspace that contains $B$, and
    $H^-=\RR^d\setminus H^+$.  The hyperplane $H$ contains $d$ points
    of $B$, and it decomposes $R$ into two convex bodies: $R\cap H^-$
    is empty and $R\gets R\cap H^+$ contains all points in $B\setminus
    H$. Update $B \gets B \setminus H$ and $R \gets R \cap H^+$.
    \item[{\sc Step 3.}] Repeat {\sc Step 2} with the new values of
    $R$ and $B$ until $B$ is the empty set.  (If $|B|< d$, then any
    supporting hyperplane of $B$ completes the partition.)
\end{enumerate}

\begin{figure}[htb]
    \centerline{\epsfxsize=6.1in \epsffile{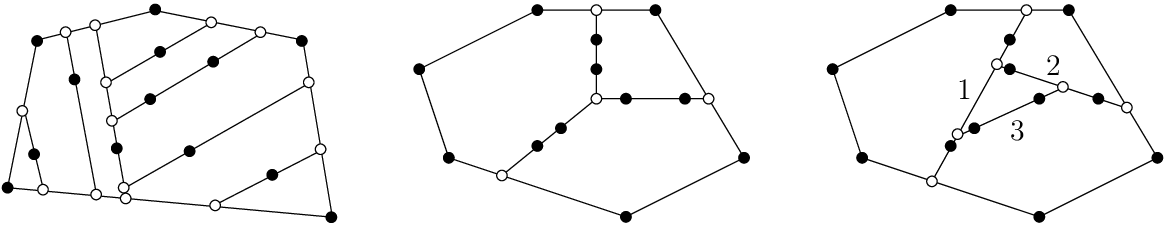}}
    \caption{Steiner convex partitions with Steiner points drawn as
       hollow circles. Left: A Steiner convex partition of a set of 13
       points.  Middle: A Steiner partition of a set of 12 points into
       three tiles.  Right: A Steiner partition of the same set of 12
       points into 4 tiles, generated by Algorithm {\bf A1} (the
       labels reflect the order of execution).}
    \figlab{f1}
\end{figure}

It is obvious that the algorithm generates a Steiner convex partition
of $S$. An illustration of Algorithm {\bf A1} on a small planar
example appears in \figref{f1}~(right).  Let $h$ and $i$ denote the
number of hull and interior points of $S$, respectively, so that
$n=h+i$. Each hyperplane used by the algorithm removes $d$ interior
points of $S$ (with the possible exception of the last round if $i$ is
not a multiple of $d$). Hence the number of convex tiles is $1+\lceil
i/d \rceil$, and we have $1+\lceil i/d \rceil = \lceil (i+d)/d\rceil
\leq \lceil (n-1)/d\rceil$, as required.

\paragraph{Lower bound in the plane.}  A matching lower
bound in the plane is given by the following construction.  For $n\geq
3$, let $S=A\cup B$, where $A$ is a set of 3 non-collinear points in
the plane, and $B$ is a set of $n-3$ points that form a regular
$(n-3)$-gon in the interior of $\conv(A)$, so that $\conv(S)=\conv(A)$
is a triangle. If $n=3$, then $\conv(S)$ is an empty triangle, and
$g_2(S)=1 =\lceil (n-1)/2\rceil$. If $4\leq n\leq 5$, $S$ is not in
convex position, and so $g_2(S)\geq 2 =\lceil (n-1)/2\rceil$.  Suppose
now that $n\geq 6$.

Consider an arbitrary convex partition of $S$. Let $o$ be a point in
the interior of $\conv(B)$ such that the lines $os$, $s\in S$, do not
contain any edges of the tiles. Refer to \figref{lower-bound}.
For each point $s\in B$, choose a \emph{reference point} $r(s)\in
\RR^2$ on the ray $\overrightarrow{os}$ in $\conv(A)\setminus
\conv(B)$ sufficiently close to point $s$, and lying in the interior
of a tile.  Note that the convex tile containing $o$ cannot contain
any reference points. We claim that any tile contains at most 2
reference points.  This immediately implies $g_2(S)\geq 1+ \lceil
(n-3)/2\rceil =\lceil (n-1)/2\rceil$.

Suppose, to the contrary, that a tile $\tau$ contains 3 reference
points $r_1,r_2,r_3$, corresponding to the points $s_1,s_2,s_3$.
Refer to \figref{lower-bound}.
\begin{figure} [htb]
\centerline{\epsfxsize=2.1in \epsffile{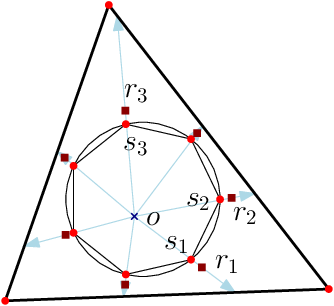}}
    \caption{Lower bound construction in $\RR^2$.}
    \figlab{lower-bound}
\end{figure}
Note that $o$ cannot be in the interior of $\tau$, otherwise $\tau$
would contain all points $s_1,s_2,s_3$ in its interior.
Hence $\conv(\{o,s_1,s_2,s_3\})$ is a quadrilateral, and
$\conv(\{o,r_1,r_2,r_3\})$ is also a quadrilateral, since the
reference points are sufficiently close to the corresponding points in
$B$.  We may assume w.l.o.g. that
vertices of $\conv(\{o,s_1,s_2,s_3\})$ are $o$, $s_1$, $s_2$, $s_3$ in
counterclockwise order. Then $s_2$ lies in the interior of
$\conv(\{r_1,r_2,r_3\})$. Hence the tile containing $r_1$, $r_2$, and $r_3$,
must contain point $s_2$ in its interior, a contradiction. We conclude that
every tile $\tau$ contains at most 2 reference points, as required.

\paragraph{Lower bounds for $d\geq 3$.} A similar
construction works in for any $d\geq 2$, but the lower bound no longer
matches the upper bound $g_d(n)\leq \lceil (n-1)/d\rceil$ for $d \geq 3$.

Recall that a \emph{Horton set}~\cite{H83} is a set $S$ of $n$ points
in the plane such that the convex hull of any 7 points is non-empty.
Valtr~\cite{Va92} generalized Horton sets to $\RR^d$. For every
$d\in \NN$, there exists a minimal integer $h(d)$ with the property
that for every $n\in \NN$ there is a set $S$ of $n$ points in general
position in $\RR^d$ such that the convex hull of any $h(d)+1$ points
in $S$ is non-empty. It is known that $h(2)=6$, and Valtr proved that
$h(3)\leq 22$, and in general that $h(d)\leq 2^{d-1}(N(d-1)+1)$, where
$N(k)$ is the product of the first $k$ primes.

We construct a set $S$ of $n\geq d+1$ points in $\RR^d$ as follows.
Let $S=A\cup B$, where $A$ is a set of $d+1$ points in general
position in $\RR^d$, and $B$ is a generalized Horton set of $n-(d+1)$
points in the interior of $\conv(A)$, such that the interior of any
$h(d)+1$ points from $B$ contains some point in $B$.

Consider an arbitrary Steiner convex partition of $S$. Every point
$b\in B$ is in the interior of $\conv(S)$, and so it lies on the
boundary of at least 2 convex tiles. For each $b\in B$, place two
\emph{reference points} in the interiors of 2 distinct tiles incident
to $b$. Every tile contains at most $h(d)$ reference points. Indeed,
if a tile contains $h(d)+1$ reference points, then it is incident to
$h(d)+1$ points in $B$, and some point of $B$ lies in the interior of
the convex hull of these points, a contradiction.

There are $2(n-d-1)$ reference points, and every tile contains at most
$h(d)$ of them. So the number of tiles is at least $\lceil
2(n-d-1)/h(d) \rceil$. Consequently, for every fixed $d\geq 2$, we
have $g_d(n) =\Omega(n)$.

\section{Approximating the minimum Steiner convex partition in $\RR^2$}
\seclab{sec:steiner}

In this section we prove~\thmref{T2} by showing that our simple-minded
algorithm {\bf A1} from \secref{sec:bounds} achieves a constant-factor
approximation in the plane if the points in $S$ are in general
position.

\paragraph{Approximation ratio.}  Recall that algorithm
{\bf A1} computes a Steiner convex partition of $\conv(S)$ into at most
$1+\lceil i/2 \rceil$ parts, where $i$ stands for the number of
interior points of $S$.

If $i=0$, the algorithm computes an optimal partition, \ie,
$\alg=\opt=1$. Assume now that $i \geq 1$.  Consider an optimal
Steiner convex partition $\Pi$ of $S$ with $\opt$ tiles.  We construct a
planar multigraph $G=(V,E)$ as follows. The \emph{faces} of $G$ are
the convex tiles and the exterior of $\conv(S)$ (the outer face).  The
\emph{vertices} $V$ are the points in the plane incident to at least 3
faces (counting the outer face as well). Since $i \geq 1$, $G$ is
non-empty and we have $|V| \geq 2$. Each \emph{edge} in $E$ is a
Jordan arc on the common boundary of two faces. An edge between two
bounded faces is a straight line segment, and so it contains at most
two interior points of $S$. An edge between the outer face and a
bounded face is a convex arc, containing hull points from $S$. Double
edges are possible if two vertices of the outer face are connected by
a straight line edge and a curve edge along the boundary---in this
case these two parallel edges bound a convex face. No loops are
possible in $G$. Since $\Pi$ is a convex partition, $G$ is connected.

Let $v$, $e$, and $f$, respectively, denote the number of vertices,
edges, and bounded (convex) faces of $G$; in particular, $f=\opt$.  By
Euler's formula for planar multigraphs, we have $v-e+f=1$, that is,
$f=e-v+1$.  By construction, each vertex of $G$ is incident to at
least 3 edges, and every edge is incident to two vertices. Therefore,
$3v \leq 2e$, or $v \leq 2e/3$.  Consequently, $f=e-v+1 \geq
e-2e/3+1=e/3+1$.  Since $S$ is in general position, each straight-line
edge of $G$ contains at most 2 interior points from $S$. Curve edges
along the boundary do not contain interior points. Hence each edge in
$E$ is incident to at most two interior points in $S$, thus $i \leq
2e$.  Substituting this into the previous inequality on $f$ yields
$\opt= f \geq e/3+1 \geq i/6 +1$.  Comparing this lower bound with the
upper bound $\alg \leq \lceil i/2\rceil+1$, we conclude that
$$ \frac{\alg}{\opt}
\leq \frac{\lceil i/2\rceil +1}{i/6+1} \leq 3 \ \frac{i+3}{i+6} <3, $$
and the approximation ratio of 3 follows.

\paragraph{Tightness of the approximation ratio.}  We
first show that the ratio $3$ established above is tight for Algorithm {\bf A1}.
We construct a planar point set $S$ as follows. Refer to  \figref{fig:tightness}~(left).
Consider a large (say, hexagonal) section of a hexagonal lattice.  Place Steiner
vertices at the lattice points, and place two points in $S$ on each
lattice edge. Slightly perturb the lattice, and add a few more points
in $S$ near the boundary, and a few more Steiner points, so as to
obtain a Steiner convex partition of $S$ with no three points
collinear. Denote by $v$, $e$, and $f$, the elements of the planar
multigraph $G$ as before. Since we consider a large lattice section,
we have $v,e,f \to \infty$. We write $a \sim b$, whenever $a/b \to 1$.
As before, we have $f+v =e+1$, and since each non-boundary edge is
shared by two convex faces, we have $e \sim 6f/2=3f$. By construction,
$i \sim 2e \sim 6f$, hence $f \sim i/6$. Therefore the convex
partition constructed above has $f \sim i/6$, while Algorithm {\bf A1}
constructs one with about $i/2$ faces. Letting $e \to \infty$, then $i
\to \infty$, and the ratio $\alg/\opt$ approaches $3$ in the limit:
$\alg/\opt \sim (i/2)/(i/6)=3$.
\medskip

\begin{figure} [htb]
\centerline{\epsfxsize=5in \epsffile{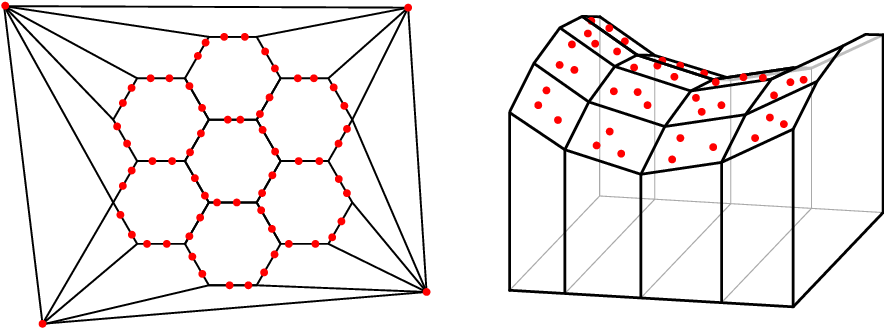}}
    \caption{Left: two points on each edge of a section of a perturbed
      hexagonal lattice in $\RR^2$, and four extra vertices of a bounding box.
     Right: Points in general position on a saddle surface in $\RR^3$.}
    \figlab{fig:tightness}
\end{figure}

\paragraph{Time analysis.} Algorithm {\bf A1} can be implemented to
run in $O(n \log n )$ time for a set $S$ of $n$ points in the plane.
We employ the semi-dynamic (delete only) convex
hull data structure of Hershberger and Suri~\cite{HS92}. This data
structure supports point deletion in $O(\log n)$ time, and uses $O(n)$
space and $O(n\log n)$ preprocessing time. We maintain the boundary of
a convex polygon $R$ in a binary search tree, a set $B\subset S$ of
points lying in the interior of $R$, and the convex hull $\conv(B)$
with the above semi-dynamic data structure~\cite{HS92}.  Initially,
$R=\conv(S)$, which can be computed in $O(n\log n)$ time; and
$B\subset S$ is the set of interior points.  In each round of the
algorithm, consider the supporting line $H$ of an arbitrary edge $e$
of $\conv(B)$ such that $B$ lies in the halfplane $H^+$.  The two
intersection points of $H$ with the boundary of $R$ can be computed in
$O(\log n)$ time. At the end of the round, we can update $B \gets B
\setminus H$ and $\conv(B)$ in $O(k\log n)$ time, where $k$ is the
number of points removed from $B$; and we can update $R \gets R \cap
H^+$ in $O(\log n)$ time. Every point is removed from $B$ exactly
once, and the number of rounds is at most $\lceil (n-3)/2\rceil$, so
the total update time is $O(n\log n)$ throughout the algorithm.

\paragraph{Remark.}
Interestingly enough, in dimensions 3 and higher, Algorithm {\bf A1}
does not give a constant-factor approximation.  For every integer $n$,
one can construct a set $S$ of $n$ points in general position in
$\RR^3$ such that $i=n-4$ of them lie in the interior of $\conv(S)$,
but the minimum Steiner convex partition has only $O(\sqrt{n})$ tiles.
In contrast, Algorithm {\bf A1} computes a Steiner partition with
$i/3=(n-4)/3$ convex tiles.

We first construct the convex tiles, and then describe the point set
$S$. Specifically, $S$ consists of 4 points of a large tetrahedron,
and 3 points in general position on the common boundary of certain
pairs of adjacent tiles.

Let $k=\lceil \sqrt{(n-4)/3}\rceil$. Place $(k+1)^2$ Steiner points
$(a,b,a^2-b^2)$ on the saddle surface $z=x^2-y^2$ for pairs of
integers $(a,b)\in \ZZ^2$, $-\lfloor k/2\rfloor \leq a,b\leq \lceil
k/2\rceil$.  The four points $\{(x,y,x^2-y^2): x\in \{a,a+1\}, y\in
\{b, b+1\}\}$ form a parallelogram for every $(a,b)\in \ZZ^2$,
% xxx: need to include a proof later, that they form a parallelogram
$-\lfloor k/2\rfloor \leq a,b\leq \lceil k/2\rceil-1$.  Refer to
\figref{fig:tightness}~(right).  These parallelograms form a terrain over
the region $\{(x,y): -\lfloor k/2\rfloor \leq x,y\leq \lceil
k/2\rceil\}$.  Note that no two parallelograms are coplanar.
Subdivide the space \emph{below} this terrain by vertical planes
$x=a$, $-\lfloor k/2\rfloor \leq a\leq \lceil k/2\rceil$.  Similarly,
subdivide the space \emph{above} this terrain by planes $y=b$,
$-\lfloor k/2\rfloor \leq b\leq \lceil k/2\rceil$. We obtain $2k$
interior-disjoint convex regions, $k$ above and $k$ below the terrain,
such that the common boundary of a region above and a region below is
a parallelogram of the terrain.  The points in $\RR^3$ that do not lie
above or below the terrain can be covered by 4 convex wedges.

Enclose the terrain in a sufficiently large tetrahedron $T$. Clip the
$2k$ convex regions and the 4 wedges into the interior of $T$. These
$2k+4$ convex bodies tile $T$.  Choose 3 noncollinear points of $S$ in
each of the $k^2$ parallelograms, such that no 4 points are coplanar
and no 2 are collinear with vertices of $T$. Let the point set $S$ be
the set of 4 vertices of the large tetrahedron $T$ and the $3k^2$
points selected from the parallelograms.

\section{Approximating the maximum empty convex body}
\seclab{sec:body}

Let $\PntSet$ be a set of points in the unit cube $[0,1]^d \subseteq
\RR^d$. Our task is to approximate the largest convex body $\Body
\subseteq [0,1]^d$ that contains no points of $\PntSet$ in its
interior. Let $\Copt = \Copt(\PntSet)$ denote this body.
% we don't need this 2nd notation!
%, and let $\vOptX{\PntSet}$ denote its volume.

\subsection{Approximation by the discrete hull}

In the following, assume that $m>0$ is some integer, and consider the
grid point set
$$ \GridX{m} = \brc{(i_1, \ldots, i_d)/m \sep{ i_1,  \ldots, i_d
\in \brc{0, 1, \ldots, m}}}. $$
Let $\PntSet \subseteq [0,1]^d$ be a point set, and let $\Copt$ be
the corresponding largest empty convex body in $[0,1]^d$.
Given a grid $\Grid(m)$, we call $\conv(\Copt \cap \GridX{m})$
the \emph{discrete hull} of $\Copt$~\cite{hp-osafd-98a}.
We need the following easy lemma.

\begin{lemma}\lemlab{discrete:hull}\label{discrete:hull}
    Let $C\subseteq [0,1]^d$ be a convex body and $D = \conv(C \cap
    \GridX{m})$.  Then we have $\VolX{C} - \VolX{ D } = O(1/m)$, where
    the constant of proportionality depends only on $d$.
\end{lemma}
\begin{proof}
    Consider a point $\pnt \in C$, and the cube $\pnt + [-2,2]^d/m$
    centered at $\pnt$ with side length $4/m$. If this cube is contained
    in $C$, then all grid points of the grid cell of $\GridX{m}$ containing
    $\pnt$ are in $C$, and $\pnt$ lies in $D$. Therefore, for every
    point $\pnt\in C\setminus D$, the cube $\pnt + [-2,2]^d/m$ is not contained
    in $C$. By convexity, at least one of the vertices of the cube
    $\pnt + [-2,2]^d/m$ lies outside of $C$. Therefore, the distance
    from $\pnt$ to the boundary of $C$ is at most the distance from $\pnt$ to
    a corner of this cube, which is $2\sqrt{d}/m$.

    It follows that all the points in the corridor $C \setminus D$ are
    at distance at most $2\sqrt{d}/m$ from the boundary of $C$. The volume
    of the boundary of $C$ is bounded from above by the volume of the boundary of
    the unit cube, namely $2d$. As such, the volume of this corridor
    is $\VolX{\bd{C}} 2\sqrt{d}/m \leq (2d) (2\sqrt{d}/m)= O(d^{3/2}/m)$.
    For a fixed $d$, this is $O(1/m)$, as claimed.
\end{proof}

\lemref{discrete:hull} implies that if $\VolX{\Copt} \geq \rho$,
in order to obtain a $(1-\eps$)-approximation,
we can concentrate our search on convex polytopes that have their
vertices at grid points in $\GridX{m}$, where $m = O(1/(\eps \rho))$.
If $\rho$ is a constant, then the maximum volume empty lattice
polytope in $\GridX{m}$ with $m=O(1/\eps)$ is an $(1-\eps)$-approximation
for $\Copt$. However, for arbitrary $\VolX{\Copt} = \Omega(1/n)$,
a much finer grid would be necessary to achieve this approximation.

\subsection{An initial brute force approach}

In this section we present approximation algorithms (for all $d$) relying on
\lemref{discrete:hull} alone, approximating the maximum volume empty
polytope by a lattice polytope in a sufficiently fine lattice (grid).
We shall refine our technique in Subsections~\ref{ssec:refine}
and~\ref{ssec:refine:d}.

For the plane, we take advantage of the existence of an efficient
solution for a related search problem. Refining a natural dynamic programming
approach by Eppstein~\etal~\cite{EORW92} and Fischer~\cite{Fis97},
Bautista-Santiago~\etal~\cite{BDL+11} obtained the following result.

\begin{lemma}[\cite{BDL+11}]\lemlab{2:d:compute:empty}\label{2:d:compute:empty}
    Given a set $\PntSet$ of $m$ points and a set $\PntSetA$ of $O(m)$
    points in the plane, one can compute a convex polygon with the
    largest area with vertices in $\PntSet$ that does not contain any
    point of $\PntSetA$ in its interior in $O(m^3)$ time.
\end{lemma}

\noindent{\bf Remark.}
The algorithm has the same running time if $\PntSetA$ is a set of
$O(m)$ forbidden rectangles.

\smallskip
The combination of Lemmas~\ref{discrete:hull} and \ref{2:d:compute:empty}
readily yields an approximation algorithm for the plane, whose running
time depends on $\VolX{\Copt}$.

\begin{lemma}\lemlab{brute:force} \label{brute:force}
    Given a set $\PntSet \subseteq [0,1]^2$ of $n$ points, such that
    $\VolX{\Copt} \geq \rho$, and a parameter $\eps>0$, one can
    compute an empty convex body $C \subseteq [0,1]^2$ such that
    $\VolX{C} \geq (1-\eps) \VolX{\Copt}$.  The running time of the
    algorithm is $O\pth{ n + 1/(\eps \rho)^6}$.
\end{lemma}
\begin{proof}
    Consider the grid $\GridX{m}$ with $m = O(1/(\eps \rho))$.
    By \lemref{discrete:hull} we can restrict our search to a grid polygon.
    Going a step further, we mark all the grid cells containing points of
    $\PntSet$ as forbidden. Arguing as in \lemref{discrete:hull},
    one can show that the area of the largest convex grid polygon avoiding
    the forbidden cells is at least $\VolX{\Copt} - c/m$, where $c$ is
    a constant.

    We now restrict our attention to the task of finding a largest
    polygon. We have a set $\PntSetA$ of $O(m^2)$ grid points that
    might be used as vertices of the grid polygon, and a set of
    $O(m^2)$ grid cells that cannot intersect the interior of the
    computed polygon. By \lemref{2:d:compute:empty}, a largest
    empty polygon can be found in $O(m^6)$ time. Setting
    $m = O(1/(\eps \rho))$, we get an algorithm with overall
    running time $O\pth{ n + 1/(\eps \rho)^6 }$.
\end{proof}

For dimensions $d\geq 3$, we are not aware of any analogue of the
dynamic programming algorithm in \lemref{2:d:compute:empty}.
Instead, we use a brute force approach that enumerates all feasible
subsets of a sufficiently fine grid.

\begin{lemma}\lemlab{brute:force:h:dim} \label{brute:force:h:dim}
    Given a set $\PntSet \subseteq [0,1]^d$ of $n$ points such that
    $\VolX{\Copt} \geq \rho$, and a parameter $\eps>0$, one can
    compute an empty convex body $C \subseteq [0,1]^d$, such that
    $\VolX{C} \geq (1-\eps)\VolX{\Copt}$. The running time of the
    algorithm is $O\pth{ n } + \exp \pth{ O\pth{m^{d(d-1)/(d+1) } \log m}}$,
    where $m = O(1/(\eps \rho))$ and $d$ is fixed.
\end{lemma}
\begin{proof}
    Consider the grid $\GridX{m}$ with $m = O(1/(\eps \rho))$.
    Let $X$ be the set of vertices of all grid cells of $\GridX{m}$
    that contain some point from $\PntSet$ (\ie, $2^d$ vertices per cell).
    Note that $\cardin{X} = O(m^d)$.
    Andrews~\cite{And63} proved that a convex lattice polytope of volume
    $V$ has $O(V^{(d-1)/(d+1)})$ vertices. Hence a convex lattice
    polytope in $\GridX{m}$ has $O(m^{d(d-1)/(d+1)})$ vertices.
     By the well-known inequality $\sum_{i=0}^k{n\choose i}\leq (\frac{en}{k})^k$,
    the number of subsets of size $O(m^{d(d-1)/(d+1)})$ from $\GridX{m}$ is
    $$\sum_{i=0}^{O(m^{d(d-1)/(d+1)})} {m^d\choose i}
    \leq \left(m^{2d/(d+1)}\right)^{O(m^{d(d-1)/(d+1)})}
    \leq \exp \left( O(m^{d(d-1)/(d+1)} \log m) \right).$$
    For each such candidate subset $G$ of size $O(m^{d(d-1)/(d+1)})$,
    test whether $\conv(G)$ is empty of points from $X$.
    For each point in $X$, the containment test reduces to a linear
    program that can be solved in time polynomial in $m$.
    Returning the subset with the largest hull volume found yields
    the desired approximation. The runtime of the algorithm is
    $\exp \left( O(m^{d(d-1)/(d+1)} \log m) \right).$
\end{proof}

B\'arany and Vershik~\cite{BV92} proved that there are
$\exp \left( O(m^{d(d-1)/(d+1)}) \right)$ convex lattice polytopes in
$\GridX{m}$. If the polytopes can also be enumerated in this time
(as in the planar case~\cite{BP92}), then the runtime
in \lemref{brute:force:h:dim} reduces accordingly.

\subsection{A faster approximation in the plane}
\label{ssec:refine}

If $\Copt$ is long and skinny (e.g., $\rho$ is close to $1/n$), then
the uniform grid $\GridX{m}$ we used in Lemmas~\ref{brute:force} and
\ref{brute:force:h:dim} is unsuitable for finding
a $(1-\eps)$-approximation efficiently. Instead, we employ a rotated and
stretched grid (an affine copy of $\GridX{m}$) that has similar
orientation and aspect ratio as $\Copt$. This overcomes one of the
main difficulties in obtaining a good approximation. Since we do not
know the shape and orientation of $\Copt$, we guess these
parameters via the minimum area rectangle containing $\Copt$.

\begin{lemma}\lemlab{algorithm:2d}%
    Given a set $\PntSet \subseteq [0,1]^2$ of $n$ points such that
    $\VolX{\Copt} \geq \rho$, and a parameter $\eps>0$,
    one can compute an empty convex body $C \subseteq [0,1]^2$ such that
    $\VolX{C} \geq (1-\eps)\VolX{\Copt}$.
   The running time of the algorithm is
$O \pth{ \rho^{-1} \pth{ n +  \rho^{-1} \eps^{-6} }}$.
% xxx: revised running time above.
% $O \pth{ \eps^{-6} \rho^{-1} \, n}$.
\end{lemma}
\begin{proof}
   The idea is to first guess a rectangle $R$ that contains
   $\Copt$ such that $\VolX{\Copt}$ is at least a constant
   fraction of the area of $\VolX{R}$, and then to apply
   \lemref{brute:force} to the rectangle $R$ (as the unit square) to
   get the desired approximation.

   Let $B_0$ be the minimum area rectangle (of arbitrary orientation)
   that contains $\Copt$; see \figref{boundingbox}~(left). We guess an
   approximate copy of $B_0$. In particular, we guess the lengths of
   the two sides of $B_0$ (up to a factor of $2$) and the orientation of
   $B_0$ (up to an angle of $O(1/n)$), and then try to position a
   scaled copy of the guessed rectangle so that that it fully contains $\Copt$.

\begin{figure} [htb]
\centerline{\epsfxsize=4.9in \epsffile{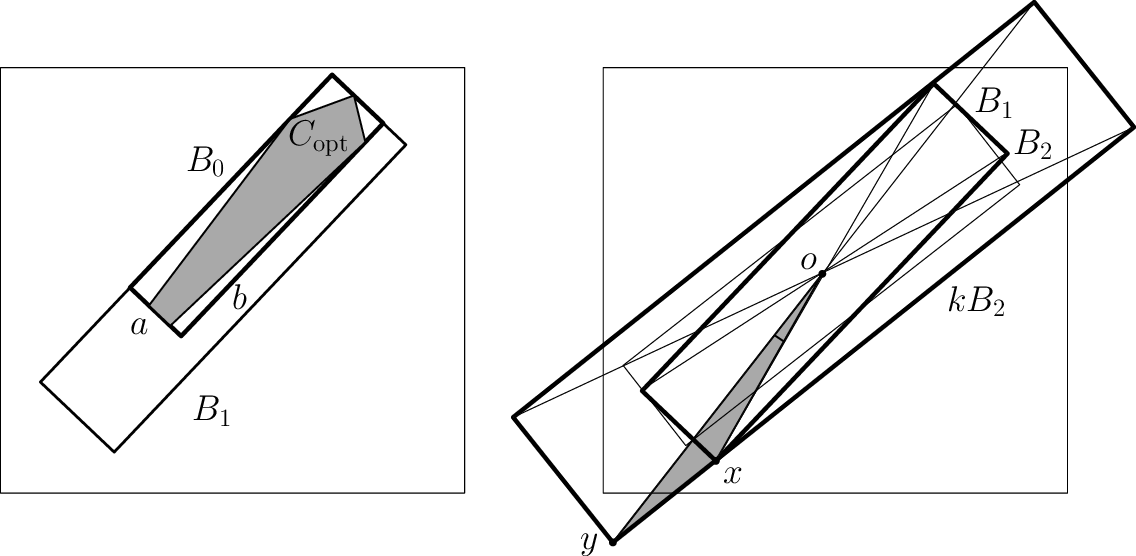}}
    \caption{Left: $\Copt$, a minimum area rectangle $B_0$,
      $\Copt\subseteq B_0$,
    and a minimum area rectangle $B_1$, $B_0\subseteq B_1$, with
    canonical side lengths and the same orientation as $B_0$.
    Right: Rectangle $B_1$, a rotated copy $B_2$ with the closest canonical
    orientation, and a minimum area scaled copy $k B_2$ such that
    $B_1 \subseteq kB_2$.}
    \figlab{boundingbox}
\end{figure}

    Assume for convenience that $n\geq 10$.
    We now show that $\VolX{\Copt}\geq \sqrt{2}/n$, using \thmref{T1}.
    Augment the point set $\PntSet$ with the four corners of the unit square $[0,1]^2$
    into a set of $n+4$ points. By \thmref{T1}, the augmented point set has a Steiner
    convex partition into at most $g_2(n+4) =\lceil \frac{n+4-1}{2}\rceil$ tiles.
    The area of the largest tile is at least that of the average tile
    in this partition, that is,
    $\VolX{\Copt}\geq 1/\lceil \frac{n+4-1}{2}\rceil \geq \frac{2}{n+4}\geq
    \frac{2}{\sqrt{2}n} =\frac{\sqrt2}{n}$, for $n\geq 10$.
    Therefore, we may assume that $\VolX{\Copt}\geq \rho \geq \sqrt{2}/n$.

    Denote by $a$ and $b$ the lengths of the two sides of $B_0$, where $a \leq b$.
    It is clear that $b \leq \sqrt{2}$, the diameter of the unit square.
    We also have $a=\VolX{B_0}/b\geq \VolX{\Copt}/b\geq \sqrt{2}/(bn)
    \geq 1/n$, hence the aspect ratio of $B_0$ is $b/a \leq \sqrt{2}/a
    \leq \sqrt{2}n$.

    Assume now that $2^{i-1}\rho \leq \VolX{\Copt} < 2^i\rho$
    for some $i=1,2,\ldots , \ceil{\log_2(\sqrt{2}n)}$.
    If we want to guess the aspect ratio of $B_0$ up to a factor of two,
    we need to consider only $O(\log \rho^{-1})$ possibilities.
    Indeed, we consider the \emph{canonical aspect ratios} $2^j$ for
    $j=0,\ldots, \ceil{\log_2 (\sqrt{2}/\rho)}-i$,
    and \emph{canonical side lengths} $2^{(i+j)/2}\sqrt{\rho}$ and
    $2^{(i-j)/2}\sqrt{\rho}$.
    Let $B_1$ be a minimum area rectangle with canonical side lengths
    and the same orientation as $B_0$, so that $B_0 \subseteq B_1$.

    The \emph{orientation} of a rectangle is given by the angle between
    one side and the $x$-axis. We approximate the orientation of $B_0$
    by \emph{canonical orientations} $\alpha=r \pi/(5\cdot 2^j)$, for
    $r =0,1,\ldots,5\cdot 2^j-1$. Let $B_2$ be a congruent copy of $B_1$
    rotated clockwise to the nearest canonical orientation about the
    center of $B_1$. We show that $B_1 \subset 2B_2$, \ie, a scaled
    copy of $B_2$ contains $B_1$. Let $k \geq 1$ be the minimum scale
    factor such that $B_1 \subseteq k B_2$. Refer to \figref{boundingbox}~(right).
    Denote by $o$ the common center of $B_1$ and $B_2$, let $x$ be a vertex
    of $B_1$ on the boundary of $kB_2$, and let $y$ be the
    corresponding vertex of $kB_2$. Clearly, $\sin(\angle xoy)\leq
    \pi/(5\cdot 2^j)$ since we rotate by at most $\pi/(5\cdot 2^j)$.
    The aspect ratio of the rectangle $kB_2$ is $\cot (\angle oyx)=2^j$.
    Since $\angle oyx<\pi/4$, we have
    $\sin(\angle oyx) = \tan (\angle oyx) \cos(\angle oyx)
    \geq 2^{-j} \cos\frac{\pi}{4} = 2^{-j-1/2} > \pi/(5\cdot 2^j)$.
    The law of sines yields $|ox|> |xy|$;
    and we have $|ox|+|xy| > |oy|$ by the triangle inequality.
    If follows that $|oy|<2|ox|$, and so $k \leq 2$ suffices.
    Summing over all possible areas, canonical aspect ratios,
    and orientations, the number of possibilities is
    $$\sum_{i=0}^{\ceil{\log_2 (\sqrt{2}/\rho)}} \sum_{j=0}^{\ceil{\log_2 (\sqrt{2}/\rho)}-i} 5\cdot 2^j
    \leq \sum_{i=0}^{\ceil{\log_2 (\sqrt{2}/\rho)}} 10\cdot 2^{\ceil{\log_2 (\sqrt{2}/\rho)}-i}
    \leq 20 \cdot 2^{\ceil{\log_2 (\sqrt{2}/\rho)}} = O(\rho^{-1}).$$

    So far we have guessed the canonical side lengths and orientation
    of $B_2$, however, we do not know its location in the plane. If a
    translated copy $B_2+v$ of $B_2$ intersects $\Copt$, then
    $3B_2+v$ contains it, since $\Copt\subseteq B_0\subseteq B_1\subseteq 2B_2$.
    Consider an arbitrary tiling of the plane with translates of $B_2$.
    By a packing argument, only $O(1/\rho)$ translates
    intersect the unit square $[0,1]^2$. One of these translates,
    say $B_2+v$, intersects $\Copt$, and hence the rectangle $R=3B_2+v$ contains $\Copt$.

    We can apply \lemref{brute:force} to the rectangle $R$ (as the
    unit square) to get the desired approximation. Specifically, let
    $T:\RR^2\rightarrow \RR^2$ be an affine transformation
    that maps $R$ into the unit square $[0,1]^2$, and apply \lemref{brute:force}
    for the point set $T(S\cap R)$ and $T(R\cap [0,1]^2)$. The grid $\GridX{m}$
    clipped in $T(R\cap [0,1]^2)$ corresponds to a stretched and rotated grid
    in $R$; each grid cell of $\GridX{m}$ is stretched to a rectangle with the
    same aspect ratio as $R$. The convex polygon $\Copt$ occupies a constant
    fraction of the area of $R$, and so the resulting running time is
    $O(n_1+1/\eps^6)$, where $n_1$ is the number of points in $R$.
    Note that the algorithm of \lemref{brute:force} partitions $R$
    into a grid with $O(1/\eps^2)$ cells. The approximation algorithm only cares
    about which cells are empty and which are not.

    Since the algorithm of \lemref{brute:force} is repeated for all
    possible positions of $R$,
    the overall running time is
    $O \pth{ \rho^{-1} \pth{ n +  \rho^{-1} \eps^{-6} }}$,
    where the first factor of $\rho^{-1}$ counts possible areas, canonical aspect ratios,
    and orientations, and the second factor of $\rho^{-1}$ inside the
    parenthesis counts possible positions of the rectangle $R$.
\end{proof}

\noindent{\bf Remark.}
If $\rho=\Omega(1)$ the running time of this planar algorithm is linear in $n$.

\medskip
Since $\rho = \Omega(1/n)$, the running time of the algorithm in \lemref{algorithm:2d}
is bounded by $O\pth{\eps^{-6} n^2}$.
We summarize our result for the plane in the following.

\medskip
\noindent
{\bf \thmref{algorithm:2d}.} {\em
    Given a set $\PntSet$ of $n$ points in $[0,1]^2$ and a
    parameter $\eps>0$, one can compute an empty convex body
    $C \subseteq [0,1]^2$, such that $\VolX{C} \geq
    (1-\eps)\VolX{\Copt}$. The running time of the algorithm is
    $O\pth{\eps^{-6} n^2}$.
}

\subsection{A faster approximation in higher dimensions}
\label{ssec:refine:d}

Given a set $\PntSet \subseteq [0,1]^d$ of $n$ points
and a parameter $\eps>0$, we compute an empty
convex body $C \subseteq [0,1]^d$ such that
$\VolX{C} \geq (1-\eps)\VolX{\Copt}$.
Similarly to the algorithm in Subsection~\ref{ssec:refine},
we guess a hyperrectangle $R$ that contains $\Copt$
such that $\VolX{\Copt}$ is at least a constant fraction
of $\VolX{R}$; and then apply \lemref{brute:force:h:dim} to $R$
(as the hypercube) to obtain the desired approximation.

Consider a hyperrectangle $B_0$ of minimum volume (and arbitrary
orientation) that contains $\Copt$. The $d$ edges incident to a vertex of a
hyperrectangle $B$ are pairwise orthogonal. We call these $d$ directions the
\emph{axes} of $B$; and the \emph{orientation} of $B$ is the set of its axes.

We next enumerate all possible discretized hyperrectangles of volume $\Omega(1/n)$,
guessing the lengths of  their axes, their orientations, and their locations
as follows:

Guess the length of every axis up to a factor of 2. Since the minimum length
of an axis in our case is $\Omega(1/n)$ and the maximum is $\sqrt{d}$,
the number of possible lengths to be considered is $O\pth{ \log^d n}$.
Let $B_1$ be a hyperrectangle of minimum volume with canonical side lengths
and the same orientation as $B_0$ such that $B_0\subseteq B_1$.

We can discretize the orientation of a hyperrectangle as follows.
We spread a dense set of points on the sphere of directions,
with angular distance $O(1/n)$ between any point on the
sphere and its closest point in the chosen set. $O(n^{d-1})$
points suffice for this purpose. We try each point as the
direction of the first axis of the hyperrectangle, and then
generate the directions of the remaining axes analogously
in the orthogonal hyperplane for the chosen direction.
Overall, this generates $O(n^{\sum_{i=1}^{d-1}i})=O(n^{d(d-1)/2})$
possibilities.

Successively replace each axis of $B_1$ by an approximate axis
that makes an angle at most $\alpha = 1/(cn)$ with its
corresponding axis, where $c=c(d)$ is a constant depending on $d$.
Let $B_2$ be a congruent copy of $B_1$ obtained in this way.
If $c=c(d)$ is sufficiently small, then $B_1\subseteq 2B_2$.

Consider a tiling of $\RR^d$ with translates of $B_2$. Note that
only $O(1/\VolX{\Copt})=O(n)$ translates intersect the unit cube
$[0,1]^d$. One of these translates $B_2+v$ intersects $\Copt$,
and then the hyperrectangle $R=3B_2+v$ contains $\Copt$.
Since $\Copt(\PntSet)$ takes a constant fraction of the volume of $R$,
we can deploy \lemref{brute:force:h:dim} in this case, and get the desired
$(1-\eps)$-approximation in
$\exp\pth{ O  \pth{\eps^{-d(d-1)/(d+1)}\log \eps^{-1}}}$ time.
Putting everything together, we obtain the following.

\medskip%
\noindent%
{\bf \thmref{empty}} {\em
    Given a set $\PntSet$ of $n$ points in $[0,1]^d$, $d\geq 3$, and
    a parameter $\eps>0$, one can compute an empty convex body
    $C \subseteq [0,1]^d$, such that $\VolX{C} \geq
    (1-\eps)\VolX{\Copt}$. The running time of the algorithm is
   $$\exp\pth{ O  \pth{\eps^{-d(d-1)/(d+1)}\log \eps^{-1}}} n^{1 + d(d-1)/2} \log^d n.$$
}

\paragraph{Remark.}  Consider a set $S$ of $n$ points in
$\RR^d$.  The approximation algorithm we have presented can be
modified to approximate the largest empty tile, \ie, the largest empty
convex body contained in $\conv(S)$, rather than $[0,1]^d$.  The
running time is slightly worse, since we need to take the boundary of
$\conv(S)$ into account. We omit the details.

\section{Conclusions} \seclab{sec:conclusion}

In this section we briefly outline two exact algorithms for finding
the largest area empty convex polygon and the largest area empty
triangle amidst $n$ points in the unit square. At the end we list a
few open problems.

\paragraph{Largest area convex polygon.}
Let $S \subset U=[0,1]^2$, where $|S|=n$. Let $T$ be the set of
four vertices of $U$. Observe that the boundary of an optimal convex
body, $\Copt$, contains at least two points from $S \cup T$.
By convexity, the midpoint of one of these $O(n^2)$ segments
lies in $\Copt$. For each such midpoint $m$,
create a weakly simple polygon $P_m$ by connecting each
point $p \in S$ to the boundary of the square along the
ray $mp$. The polygon $P_m$ has $O(n)$ vertices and is empty of points
from $S$ in its interior. Then apply the algorithm of Chang and
Yap~\cite{CY86} for the potato-peeling problem (mentioned
in~\secref{sec:intro}) in these $O(n^2)$ weakly simple
polygons. The algorithm computes a largest area empty convex
polygon contained in a given (non-convex) polygon with $n$ vertices in
$O(n^7)$ time. Finally, return the largest convex polygon obtained
in this way. The overall running time is $O(n^9)$.

The running time can be reduced to $O(n^8 \log n)$ as follows.
Instead of considering the $O(n^2)$ midpoints, compute a set $P$ of
$O(n \log n)$ points so that every convex set of area at least
$2/(n+4)$ contains at least one of these points. In particular, $\Copt$
contains a point from $P$. The set $P$ can be computed by starting with
a $O(n) \times O(n)$ grid, and then computing an $\eps$-net for it,
where $\eps=O(1/n)$, using discrepancy~\cite{Ma02}. The running
time of this deterministic procedure is roughly $O(n^2)$, and the
running time of the overall algorithm improves to
$O(n^7 \cdot n \log{n})= O(n^8 \log n)$.

\paragraph{Largest area empty triangle.}
The same reduction can be used for finding largest area empty triangle
contained in $U$, resulting in $O(n^2)$ weakly simple polygons $P_m$.
Then the algorithm of Melissaratos and Souvaine~\cite{ms-sphsg-92} for
finding a largest area triangle contained in a polygon is applied to
each of these $O(n^2)$ polygons.
The algorithm finds such a triangle in $O(n^4)$ time, given a
polygon with $n$ vertices. Finally, return the largest triangle
obtained in this way. The overall running time is $O(n^6)$. Via the
$\eps$-net approach (from the previous paragraph) the running time of
the algorithm improves to $O(n^4 \cdot n \log{n})= O(n^5 \log n)$.

\paragraph{Open questions.}
Interesting questions remain open regarding the structure of optimal
Steiner convex partitions and the computational complexity of
computing such partitions. Other questions relate to the problem of
finding the largest empty convex body in the presence of points.

\begin{itemize}

    \item [(1)] Is there a polynomial-time algorithm for computing a
    minimum Steiner convex partition of a given set of $n$ points in
    $\RR^d$? Is there one for points in the plane?

    \item [(2)] Is there a constant-factor approximation algorithm for
    the minimum Steiner convex partition of an arbitrary point set in
    $\RR^d$ (without the general position restriction)? Is there one
    for points in the plane?

    \item [(3)] For $d>2$, the running time of our approximation
    algorithm for the maximum empty polytope has a factor of the form
    $n^{O(d^2)}$.  It seems natural to conjecture that this term can
    be reduced to $n^{O(d)}$. Another issue of interest is extending
    Lemma~\ref{2:d:compute:empty} to higher dimensions for a faster
    overall algorithm.

    \item [(4)] Given $n$ points in $[0,1]^d$, the problem of finding
    the largest convex body in $[0,1]^d$ that contains up to $k$
    (outlier) points naturally suggests itself and appears to be also
    quite challenging.

\end{itemize}

\paragraph{Acknowledgement.}
The authors thank Joe Mitchell for helpful discussions regarding the
exact algorithms in~\secref{sec:conclusion}, in particular for
suggesting the reduction of the maximum-area-empty-convex-body problem
to the potato-peeling problem. Many thanks also go to Sergio Cabello
and Maria Saumell for pointing us to the recent results of
Bautista-Santiago~\etal~\cite{BDL+11} and for suggesting logarithmic factor
improvements in the running time of the approximation algorithm in
Section~\ref{ssec:refine}.

\end{document}